\documentclass[11pt]{article}

\usepackage[letterpaper, margin=1in]{geometry}

\usepackage[parfill]{parskip}

\usepackage[utf8]{inputenc}

\usepackage[ruled, vlined, linesnumbered]{algorithm2e}

\usepackage{varioref}
\usepackage[linkcolor=black,colorlinks=true,citecolor=black,filecolor=black]{hyperref}
\usepackage{comment}

\usepackage{mathtools}
\usepackage{amsmath, amssymb, amsfonts, mathrsfs}
\usepackage[sc]{mathpazo} 
\usepackage{amsthm}

\usepackage{subcaption}
\usepackage{aligned-overset}
\usepackage{cleveref}
\usepackage{thmtools}
\usepackage{thm-restate}
\usepackage{nicefrac}
\usepackage{orcidlink}

\DeclareUnicodeCharacter{2229}{$\cap$}

\theoremstyle{plain}
\numberwithin{equation}{section}
\newtheorem{theorem}{Theorem}[section]

\newtheorem{definition}[theorem]{Definition}
\newtheorem{lemma}[theorem]{Lemma}

\newtheorem{observation}[theorem]{Observation}

\crefname{corollary}{Corollary}{Corollaries}
\crefname{lemma}{Lemma}{Lemmata}

\newcommand{\N}{\mathbb{N}}

\newcommand{\bigO}{\mathcal{O}}

\renewcommand{\epsilon}{\ensuremath\varepsilon}

\renewcommand{\phi}{\ensuremath{\varphi}}

\DeclareMathOperator{\GLB}{GLB}
\DeclareMathOperator{\LUB}{LUB}
\DeclareMathOperator{\dia}{diam}
\newcommand{\up}[1]{u^{(#1)}}
\newcommand{\down}[1]{d^{(#1)}}

\newcommand{\SUTarski}{\textsc{Super-Unique-Tarski}}
\newcommand{\UTarski}{\textsc{Unique-Tarski}}

\newcommand{\Tarski}{\textsc{Tarski}}

\usepackage{authblk}
\author{Sebastian Haslebacher\orcidlink{0000-0003-3988-3325}}
\author{Jonas Lill\orcidlink{0009-0007-1448-2425}}
\affil{Department of Computer Science, ETH Zurich, Switzerland}
\title{A Levelset Algorithm for 3D-\Tarski}
\date{}

\begin{document}

\maketitle

\begin{abstract}
    We present a simple new algorithm for finding a Tarski fixed point of a monotone function $F : [N]^3 \rightarrow [N]^3$. Our algorithm runs in $\bigO(\log^2 N)$ time and makes $\bigO(\log^2 N)$ queries to $F$, matching the $\Omega(\log^2 N)$ query lower bound due to Etessami et al.\ as well as the existing state-of-the-art algorithm due to Fearnley et al.
\end{abstract}

\newpage
\section{Introduction}

We say that a function $F : [N]^d \rightarrow [N]^d$ is monotone if $x \leq y$ implies $F(x) \leq F(y)$ for all $x, y \in [N]^d$, where $\leq$ is the partial order obtained by coordinate-wise comparison. Tarski~\cite{tarskiLatticetheoreticalFixpointTheorem1955} proved that any such monotone function $F$ must have a fixed point $x^\star = F(x^\star)$. Finding such a fixed point algorithmically has important applications in algorithmic game theory: For example, finding equilibria of supermodular games, approximating the value of Condons's or Shapley's stochastic games, and solving a zero-player game known as ARRIVAL all reduce to finding such a Tarski fixed point~\cite{dangComputationsComplexitiesTarskis2025, etessamiTarskiTheoremSupermodular2020, gartnerSubexponentialAlgorithmARRIVAL2021}. 

More concretely, given query access to a function $F : [N]^d \rightarrow [N]^d$, we refer to the computational problem of finding a fixed point of $F$ as \Tarski. A recursive binary search algorithm due to Dang, Qi, and Ye~\cite{dangComputationsComplexitiesTarskis2025} can solve this problem in $\bigO(\log^d N)$ time and queries for fixed $d$, and initially it was rather unclear whether one should expect much better algorithms. Indeed, Etessami, Papadimitriou, Rubinstein, and Yannakakis~\cite{etessamiTarskiTheoremSupermodular2020} proved a query lower bound of $\Omega(\log^2 N)$ for $d \geq 2$ and conjectured a lower bound of $\Omega(\log^d N)$ in general. However, Fearnley, Pálvölgyi, and Savani~\cite{fearnleyFasterAlgorithmFinding2022} later disproved this conjecture by giving an algorithm solving the case $d = 3$ in $\bigO(\log^2 N)$ queries and time. They also proved a decomposition lemma that allowed them to use recursive applications of their 3D-algorithm to get an $\bigO(\log^{\lceil \nicefrac{2d}{3} \rceil} N)$-algorithm in general.

Since then, both the upper bounds and lower bounds have been improved, but the true complexity of the problem remains wide open. Concretely, the current best upper bound is due to Chen and Li~\cite{chenImprovedUpperBounds2022} who proved that a fixed point can be found in  $\bigO(\log^{\lceil \nicefrac{d + 1}{2} \rceil } N)$ time and queries. In terms of lower bounds, the state-of-the-art is an $\Omega(\nicefrac{d \log^2 N}{\log d})$ query lower bound for $d \geq 2$ that was recently announced by Brânzei, Phillips, and Recker~\cite{branzeiTarskiLowerBounds2025}.

In this paper, we present a new algorithm that can solve \Tarski\ for $d = 3$ in $\bigO(\log^2 N)$ time and queries, matching the asymptotically optimal algorithm by Fearnley, Pálvölgyi, and Savani~\cite{fearnleyFasterAlgorithmFinding2022}. Their algorithm proceeds by binary searching over two-dimensional slices of the 3D-grid. Their key insight is that it is possible to find a progress point $x$ satisfying either $F(x) \leq x$ or $F(x) \geq x$ in a given two-dimensional slice $S_k = \{x \in [N]^3 \mid x_1 = k\}$ in $\bigO(\log N)$ time. This inner algorithm is however quite complicated and in particular involves an intricate case analysis. We avoid this by binary searching over levelsets instead of slices. Concretely, given a number $k \in \{3, \dots, 3N\}$, we show how we can either find a point $x$ with $F(x) \leq x$ in $L_{\leq k} = \{x \in [N]^3 \mid x_1 + x_2 + x_3 \leq k \}$ or a point $x$ with $F(x) \geq x$ in $L_{\geq k} = \{x \in [N]^3 \mid x_1 + x_2 + x_3 \geq k \}$ in $\bigO(\log N)$ time. We achieve this by searching over the levelset $L_k = \{x \in [N]^3 \mid x_1 + x_2 + x_3 = k \}$. The advantage of levelsets is that their symmetry allows our algorithm to treat all three dimensions exactly the same. However, in contrast to the slices in~\cite{fearnleyFasterAlgorithmFinding2022}, no two points in a given levelset are comparable under the partial order $\leq$, making it seemingly more difficult to exploit monotonicity of $F$. Our main contribution is to overcome this difficulty.

Naturally, our algorithm can also be plugged into the decomposition lemma of~\cite{fearnleyFasterAlgorithmFinding2022} to obtain an $\bigO(\log^{\lceil \nicefrac{2d}{3} \rceil} N)$ bound for general $d$. Moreover, our levelset approach can in principle also be applied directly to higher-dimensional instances. Unfortunately, we have not yet been able to get close to the $\bigO(\log^{\lceil \nicefrac{2d}{3} \rceil} N)$ bound (or the state-of-the-art $\bigO(\log^{\lceil \nicefrac{d + 1}{2} \rceil } N)$ bound of Chen and Li~\cite{chenImprovedUpperBounds2022}, for that matter) for $d > 3$ using this more direct approach, but we do think that this is an interesting idea for future work. 

Finally, we want to point out that our levelset approach seems even more appropriate in restricted versions of \Tarski\ such as \UTarski\ and \SUTarski. Indeed, many of our ideas originated from studying \SUTarski. In terms of query complexity, studying \SUTarski\ would have been enough, given that Chen, Li, and Yannakakis~\cite{chenReducingTarskiUnique2023} proved that \Tarski, \UTarski, and \SUTarski\ all have the exact same query complexity. Unfortunately, their reductions are not time-efficient, which is why we avoided using them and instead generalized our ideas to work for \Tarski\ directly. 

\section{Preliminaries}

We use a slightly generalized definition of \Tarski\ on an integer grid $G = [n_1] \times [n_2] \times \dots \times [n_d]$, for arbitrary $n_1, \dots, n_d \in \N$. We will use $N = \sum_{i = 1}^d n_i$ to measure the size of $G$. The grid $G$ can be equipped with the partial order $\leq$ defined as 
\[
    x \leq y \iff \forall i \in [d] : x_i \leq y_i 
\]
for all $x, y \in G$. The tuple $(G, \leq)$ forms a lattice, but we will usually just (implicitly) denote this lattice by $G$. We use $\GLB(x^{(1)}, \dots, x^{(k)})$ to denote the greatest lower bound (meet) and $\LUB(x^{(1)}, \dots, x^{(k)})$ to denote the least upper bound (join) of a set of points $x^{(1)}, \dots, x^{(k)} \in G$ in this lattice. We also use $|x| \coloneqq \sum_{i = 1}^d x_i$ for the $1$-norm of $x \in G$. A function $F : G \rightarrow G$ is called monotone (or non-decreasing) if and only if 
\[
    x \leq y \implies F(x) \leq F(y)
\]
for all $x, y \in G$. The computational problem \Tarski\ is defined as follows: Given query access to a monotone function $F : G \rightarrow G$, find a fixed point $x^\star = F(x^\star)$ of $F$. Recall that such a fixed point has to exist by Tarski's theorem~\cite{tarskiLatticetheoreticalFixpointTheorem1955}. 

Our algorithm will take $\bigO(\log^2 N)$ time and will make $\bigO(\log^2 N)$ queries to $F$. Note that this excludes the time needed to actually make a query to $F$, which naturally depends on how access to $F$ is provided to us. For any specific model of access to $F$, one would have to multiply our bounds by the access time to $F$ in the given model. Moreover, our runtime is measured in arithmetic operations: Numbers in our algorithms can be as large as $\Omega(N)$, which means that they require $\Omega(\log N)$ bits to represent, and this would have to be factored in for other models of computation. We decided to abstract these concerns away for the sake of simplicity. In particular, we will from now on simply refer to the runtime of the algorithm, factoring out any overhead that is needed to actually evaluate $F$ or to calculate with $\bigO(\log N)$-bit numbers in certain models.

Moreover, note that our presentation focuses on the setting where $F$ is promised to be monotone. However, it is not hard to see that our algorithm could be adapted to the total setting, where the algorithm is allowed to return a violation of monotonicity instead of a fixed point whenever $F$ is not monotone. This is because our algorithm only exploits monotonicity locally: If a certain step of the algorithm should fail, then a violation of monotonicity must be present among a constant set of previous queries.

We will frequently use the notion of upward and downward points: A point $x \in G$ is called upward if $F(x) \geq x$, and downward if $F(x) \leq x$. Note that a point is both upward and downward simultaneously if and only if it is a fixed point of $F$. Also note that $(1, 1, \dots, 1) \in G$ must be an upward and $(n_1, n_2, \dots, n_d) \in G$ a downward point. Furthermore, any upward point $x$ and downward point $y$ satisfying $x \leq y$ imply that there exists a fixed point $x \leq x^\star \leq y$ of $F$: Indeed, due to monotonicity of $F$, the restriction of $F$ to the subgrid $G' = \{ z \in G \mid x \leq z \leq y\}$ is itself a valid \Tarski-instance. In other words, finding such $x$ and $y$ allows us to reduce the search to $G'$.

\section{Levelset Algorithm}
\label{sec:levelset_algorithm}

We start the presentation of our algorithm with the high-level idea and then proceed by explaining all the details step-by-step. The following notion of levelsets will be crucial.

\begin{definition}[Levelset] 
\label{def:levelsets}
    Let $G$ be a $d$-dimensional integer grid of size $N = \sum_{i = 1}^d n_i$ and let $k \in \{d, d + 1, \dots, N\}$ be arbitrary. The $k$-th levelset $L_k$ of $G$ is defined as $L_k \coloneqq \{ x \in G \mid |x| = k \}$. Similarly, we define $L_{\geq k} \coloneqq \bigcup_{i \geq k} L_i$ and $L_{\leq k} \coloneqq \bigcup_{i \leq k} L_i$.
\end{definition}

\begin{figure}[htb]
    \centering
    \includegraphics[width=0.7\linewidth]{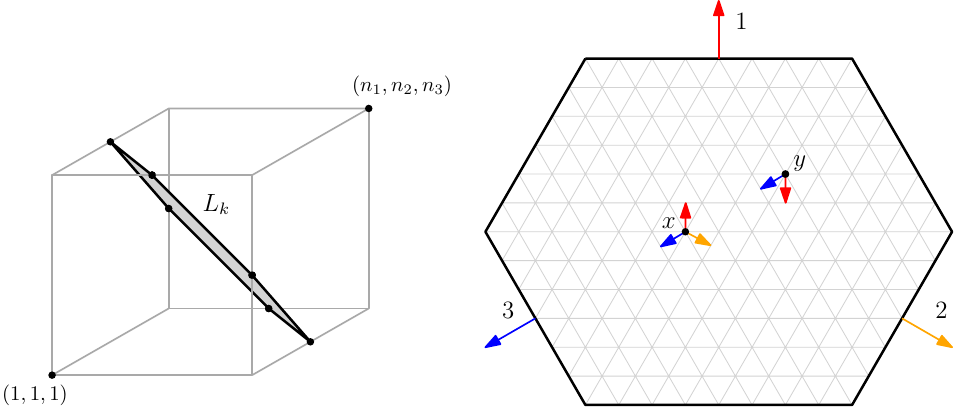}
    \caption{Levelset $L_k$ of a 3D-instance (on the left) projected to 2D (on the right). On the right, the three dimensions are indicated in red, orange, and blue, respectively. For both points $x$ and $y$, we indicate the function $F$ in each dimension separately. For example, $x$ is an upward point with $F(x)_i > x_i$ for all $i \in [3]$ while $y$ satisfies $F(y)_1 < y_1, F(y)_3 > y_3$, and $F(y)_2 = y_2$ (indicated by the absence of an orange arrow). }
    \label{fig:levelset}
\end{figure}

From now on, we restrict ourselves to the case $d = 3$. Our algorithm proceeds by finding an upward point in $L_{\geq \lceil N / 2 \rceil }$ or a downward point in $L_{\leq \lceil N / 2 \rceil }$ in time $\bigO(\log N)$. Once we have found such a point, we can shrink our search space by at least half and recursively apply the algorithm on the remaining subgrid, yielding an overall runtime of $\bigO(\log^2 N)$.

\begin{lemma}
\label{lemma:high-level_algo}
    Assume that for an arbitrary levelset $L_k$ in an instance $F : G \rightarrow G$ of 3D-\Tarski, we can find an upward point $x \in  L_{\geq k}$ or a downward point $x \in L_{\leq k}$ in time $\bigO(\log N)$. Then, a fixed point of $F$ can be found in time $\bigO(\log^2 N)$.
\end{lemma}
\begin{proof}
    We start by calling the subprocedure on the levelset $L_k$ for $k = \lceil N / 2 \rceil$ in order to find a point $x$ that is either an upward point in $L_{\geq k}$ or downward point in $L_{\leq k}$ . If $x$ is an upward point in $L_{\geq k}$, we recurse on the subgrid $G' = \{ y \in G \mid x \leq y \}$. Otherwise, $x$ is a downward point in $L_{\leq k}$ and we recurse on the subgrid $G' = \{ y \in G \mid x \geq y \}$. In both cases, $F$ restricted to $G'$ is again a \Tarski-instance. Moreover, the subgrid $G'$ has size at most half of $G$ (the parameter $N$ has shrunk by at least half). Once we arrive at a grid of constant size, we can find a fixed point in constant time. Thus, the overall procedure takes $\bigO(\log^2 N)$ time.
\end{proof}

Given \Cref{lemma:high-level_algo}, the main challenge that we have to overcome is to find either an upward point in $L_{\geq k}$ or a downward point in $L_{\leq k}$ in time $\bigO(\log N)$. For this, we will maintain a search space spanned by six bounding points within $L_k$. We will then iteratively shrink this search space in a binary search fashion. If we do not find an upward or downward point directly during the shrinking, we eventually arrive at a situation that must imply the existence of one out of three possible configurations among the six bounding points. For each of those three configurations, we then show how it implies an upward or downward point, as desired.

\subsection{Shrinking the Search Space}

We start by defining our search space and explaining how we can iteratively shrink it. For this, we first need to introduce so-called $i$-upward and $i$-downward points.

\begin{definition}[$i$-Upward and $i$-Downward Points]
    We call a point $x \in L_k$ $i$-upward for some $i \in [3]$ if 
    $F(x)_i > x_i$ and $F(x)_j \leq x_j$ for all $j \in [3] \setminus \{i\}$.
    Symmetrically, we call $x$ $i$-downward if 
    $F(x)_i < x_i$ and
     $F(x)_j \geq x_j$ for all $j \in [3] \setminus \{i\}$.
\end{definition}

\begin{figure}[htb]
    \centering
    \includegraphics[width=0.4\linewidth]{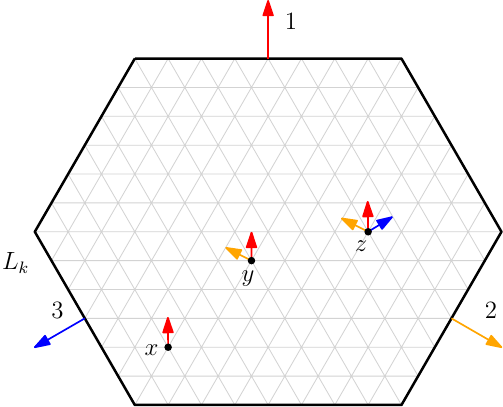}
    \caption{The three points $x, y, z$ are examples of $1$-upward points. Note that $x$ is in fact an overall upward point since it satisfies $F(x)_i \geq x_i$ for all $i \in [3]$. }
    \label{fig:1-upward}
\end{figure}

Note that any point that is neither upward nor downward must be $i$-upward or $i$-downward for some $i \in [3]$. This observation will be crucial for us and only holds for $d = 3$, which is the main reason why we could not yet generalize our algorithm to higher dimensions.

We now proceed to define our search space: At each step of the shrinking algorithm, our remaining search space is induced by six bounding points. Concretely, for all $i \in [3]$, we will maintain an $i$-upward point $\up{i} \in L_k$ and an $i$-downward point $\down{i} \in L_k$ satisfying $\up{i}_i \leq \down{i}_i$. 

\begin{definition}[Remaining Search Space] 
    Given the bounding points $\up{i}, \down{i} \in L_k$ satisfying $\up{i}_i \leq \down{i}_i$ for all $i \in [3]$, the induced remaining search space is defined as
    \[
        S \coloneqq \{ x \in L_k \mid \up{i}_i \leq x_i \leq \down{i}_i \text{ for all } i \in [3]\}
    \]
    and its diameter is a 3D-vector defined as 
    \[
        \dia(S)_i \coloneqq \max_{x, y \in S} |x_i - y_i| = \max_{x \in S} x_i - \min_{y \in S} y_i
    \]
    for all $i \in [3]$.
\end{definition}

\begin{figure}[hbt]
    \centering
    \includegraphics[width=0.4\linewidth]{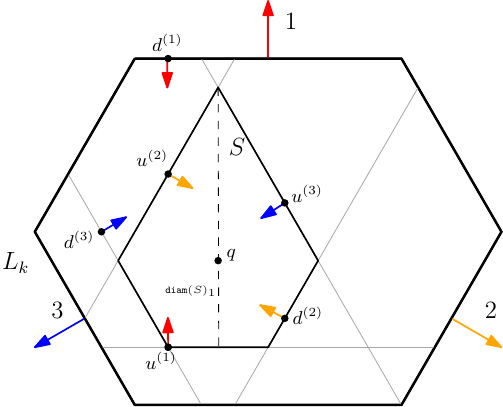}
    \caption{Example of the remaining search space $S$ (the region bounded by the black pentagon), induced by the six drawn bounding points. For simplicity, we only drew the arrow in dimension $i$ for each point $\up{i}$ or $\down{i}$. Observe in particular how $\down{1}$ does not currently contribute to the boundary of $S$, and in particular we have $\dia(S)_1 < \down{1}_1  - \up{1}_1$. Still, there is a point $q$ such that querying $q$ guarantees significant progress: Either $q$ is upward or downward which allows us to terminate, or it is $i$-upward or $i$-downward for some $i \in [3]$ which allows us to considerably shrink the diameter in one of the three dimensions.}
    \label{fig:enter-label}
\end{figure}

In the remainder of this section, we will prove that in $\bigO(\log N)$ time, we can shrink the remaining search space until we have $\dia(S)_i \leq 1$ for at least one $i \in [3]$. The proof of this is split into \Cref{lemma:shrink_search_space_I} and \Cref{lemma:shrink_search_space_II}.

\begin{lemma}[Shrinking Remaining Search Space I]
\label{lemma:shrink_search_space_I}
    Assume that we have $\dia(S)_i > 1$ for all $i \in [3]$ and $\max_{i \in [3]} \dia(S)_i \geq 6$ for our remaining search space $S$. Then there exists $q \in S$ such that 
    \[
        \max \left( \dia(\{ x \in S \mid x_i \leq q_i \})_i, \, \dia( \{ x \in S \mid x_i \geq q_i \})_i \right)  \leq \frac{5}{6} \dia(S)_i 
    \]
    for all $i \in [3]$.
\end{lemma}
\begin{proof}
    Assume without loss of generality $\dia(S)_1 \leq \dia(S)_2 \leq \dia(S)_3$. Let $\ell_i \coloneqq \min_{x \in S} x_i$ and analogously $r_i \coloneqq \max_{x \in S} x_i$ for $i \in [3]$. To prove the statement, it suffices to prove that there exists $q \in L_k$ satisfying
    \[
        \left( \ell_1 + \lceil \frac{1}{6}\dia(S)_1 \rceil, \ell_2 + \lceil \frac{1}{6}\dia(S)_2 \rceil, \ell_3 + \lceil \frac{1}{6}\dia(S)_3 \rceil \right) \leq q 
    \]
    and 
    \[
        \left( r_1 - \lceil \frac{1}{6}\dia(S)_1 \rceil , r_2 - \lceil \frac{1}{6}\dia(S)_2 \rceil, r_3 - \lceil \frac{1}{6}\dia(S)_3 \rceil \right) \geq q.
    \]
    To see that such a $q$ exists, observe first that we must have 
    \[
        \ell_1 + \lceil \frac{1}{6}\dia(S)_1 \rceil + \ell_2 + \lceil \frac{1}{6}\dia(S)_2 \rceil + \ell_3 + \lceil \frac{1}{6}\dia(S)_3  \rceil \leq k 
    \]
    because we can upper bound $\dia(S)_1$ and $\dia(S)_2$ by $\dia(S)_3$, use $\dia(S)_3 \geq 6$, and use that
    \[
        \ell_1 +  \ell_2 + \ell_3 + \dia(S)_3 \leq k,
    \] 
    by definition of $\dia(S)_3$ and $\ell_1, \ell_2, \ell_3$. Analogously, we derive
    \[
        k \leq r_1 - \lceil \frac{1}{6}\dia(S)_1 \rceil + r_2 - \lceil \frac{1}{6}\dia(S)_2 \rceil + r_3 - \lceil \frac{1}{6}\dia(S)_3  \rceil. 
    \]
\end{proof}

Note that the point $q$ in \Cref{lemma:shrink_search_space_I} can be constructed efficiently from the six bounding points. Indeed, it suffices to compute the values $\ell_i, r_i, \dia(S)_i$ for all $i \in [3]$ and then choose $q$ with $|q| = k$ arbitrarily within the described bounds. 

Now recall that \Cref{lemma:shrink_search_space_I} allows us to shrink the remaining search space only as long as we have $\max_{i \in [3]} \dia(S)_i \geq 6$. Thus, we still have to handle the case $\max_{i \in [3]} \dia(S)_i < 6$. For this, we exploit that $\max_{i \in [3]} \dia(S)_i < 6$ implies that $|S|$ is constant.

\begin{lemma}[Shrinking Remaining Search Space II]
\label{lemma:shrink_search_space_II}
    Assume that we have $\dia(S)_i > 1$ for all $i \in [3]$ and assume $S$ has constant size. Then we can, in constant time, either decrease $\dia(S)_i$ for at least one $i \in [3]$ or find an upward point in $L_{\geq k}$ or a downward point in $L_{\leq k}$.
\end{lemma}
\begin{proof}
    Consider again the coordinates $\ell_i \coloneqq \min_{x \in S} x_i$ and $r_i \coloneqq \max_{x \in S} x_i$ for $i \in [3]$. If there exists $q \in S$ with $q_i \notin \{\ell_i, r_i\}$ for all $i \in [3]$, then querying that point $q$ will allow us to shrink some diameter, as desired. Thus, assume now that no such point exists. In particular, there is no $q \in S$ satisfying
    \[
        \left( \ell_1 + 1, \ell_2 + 1, \ell_3 + 1 \right) \leq q \leq \left( r_1 - 1, r_2 - 1, r_3 - 1 \right).
    \]
    Hence, we must have $\ell_1 + 1 + \ell_2 + 1 + \ell_3 + 1 > k$ or $r_1 - 1 + r_2 - 1 + r_3  - 1 < k$. Without loss of generality, assume the former (the other case is symmetric). Note that we can actually deduce $\ell_1 + 1 + \ell_2 + 1 + \ell_3 + 1 = k + 1$ by observing that we must have $\ell_1 + \ell_2 + r_3 \leq k$ and $\ell_3 + 2 \leq r_3$. Consider the point $q^{(1)} = (\ell_1, \ell_2 + 1, \ell_3 + 1) \in S$. If it is not $1$-upward, we are guaranteed to make progress: Indeed, any of the other options allows us to make progress by decreasing some diameter or directly solving the levelset. Analogously, we make progress if $q^{(2)} = (\ell_1 + 1, \ell_2, \ell_3 + 1) \in S$ is not $2$-upward or if $q^{(3)} = (\ell_1 + 1, \ell_2 + 1, \ell_3) \in S$ is not $3$-upward. It remains to observe that if $q^{(i)}$ is $i$-upward for all $i \in [3]$, then the point $\LUB(q^{(1)}, q^{(2)}, q^{(3)})$ is upward. 
\end{proof}

Together, \Cref{lemma:shrink_search_space_I} and \Cref{lemma:shrink_search_space_II} imply that we can shrink our remaining search space until we reach $\dia(S)_i \leq 1$ for some $i \in [3]$ in time $\bigO(\log N)$, or find an upward or downward point on the way. 

\subsection{Three Configurations}
\label{sec:configurations}

Assume now that we have reached the situation $\dia(S)_i \leq 1$ for some $i \in [3]$. In that case, we will prove in \Cref{lemma:implied_configuration} that we can find one of three configurations of points, each of which leads us to an upward point in $L_{\geq k}$ or downward point in $L_{\leq k}$. 

\begin{definition}[First Configuration]
    We say that two points $x, y \in L_k$ are in first configuration if for some distinct $i, j \in [3]$, we have that
    \begin{itemize}
        \item $x$ is $i$-upward, $y$ is $j$-upward, $x_i \geq y_i$, and $x_j \leq y_j$, 
        \item or symmetrically, $x$ is $i$-downward, $y$ is $j$-downward, $x_i \leq y_i$, and $x_j \geq y_j$.
    \end{itemize}
\end{definition}

\begin{definition}[Second Configuration]
    We say that three points $x, y, z \in L_k$ are in second configuration if for some distinct $i, j, p \in [3]$, we have that
    \begin{itemize}
        \item $x$ is $i$-upward, $y$ is $j$-upward, $z$ is $p$-upward, $x_i \geq y_i$, $y_j \geq z_j$, and $z_p \geq x_p$, 
        \item or symmetrically, $x$ is $i$-downward, $y$ is $j$-downward, $z$ is $p$-downward, $x_i \leq y_i$, $y_j \leq z_j$, and $z_p \leq x_p$.
    \end{itemize}
\end{definition}

Recall that an $i$-upward point $x$ must satisfy $F(x)_j \leq x_j$ for all $j \in [3] \setminus \{i\}$. So if, for example, we have $x, y$ in first configuration with $x$ being $i$-upward and $y$ being $j$-upward, then we can deduce that $\GLB(x, y)$ is downward, by monotonicity of $F$. In fact, this argument works for any kind of first or second configuration, yielding the following two observations.

\begin{observation}[First Configuration $\implies$ Progress]
\label{observation:first_config}
    Assume that $x, y \in L_k$ are in first configuration. Then either $\GLB(x, y) \in L_{\leq k}$ is downward or $\LUB(x, y) \in L_{\geq k}$ is upward.
\end{observation}

\begin{observation}[Second Configuration $\implies$ Progress]
\label{observation:second_config}
    Assume that $x, y, z \in L_k$ are in second configuration. Then either $\GLB(x, y, z) \in L_{\leq k}$ is downward or $\LUB(x, y, z) \in L_{\geq k}$ is upward.
\end{observation}

\begin{definition}[Third Configuration]
\label{def:opposite_points}
    We consider two points $x, y \in L_k$ to be in third configuration if there is $i \in [3]$ such that $x$ is $i$-upward, $y$ is $i$-downward, and $x_i \leq y_i \leq x_i + 1$.
\end{definition}

In contrast to the first and second configuration, the third configuration does not directly imply an upward or downward point. However, we will see in \Cref{sec:third_configuration} how we can still make progress using at most $\bigO(\log N)$ additional time. But before getting to that, let us prove that a search space $S$ with $\dia(S)_i \leq 1$ for some $i \in [3]$ implies existence of at least one of the three configurations. 

\begin{figure}[htb]
    \centering
    \includegraphics[width=0.8\linewidth]{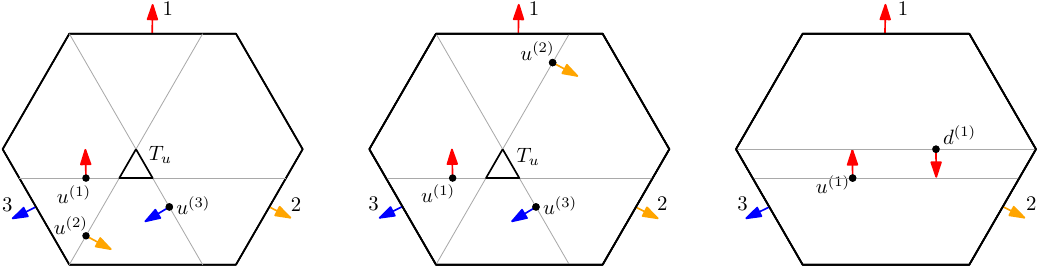}
    \caption{Illustrations for \Cref{lemma:implied_configuration}. On the left, the two points $\up{1}$ and $\up{2}$ are in first configuration. In the middle, the three points $\up{1}, \up{2}, \up{3}$ are in second configuration. On the right, the two points $\up{1}, \down{1}$ are in third configuration. For simpilicity, we only show the arrow in dimension $i$ for points $\up{i}$ and $\down{i}$.}
    \label{fig:three_configurations}
\end{figure}

\begin{lemma}[Small Diameter $\implies$ Configuration]
\label{lemma:implied_configuration}
    Assume $\up{i}_i \leq \down{i}_i$ for all $i \in [3]$, and assume that the induced remaining search space $S$ has diameter $\dia(S)_i \leq 1$ for some $i \in [3]$. Then a first, second, or third configuration can be found among the six points $\up{1}, \down{1}, \up{2}, \down{2}, \up{3}, \down{3}$. 
\end{lemma}
\begin{proof}
    Consider the two triangle-shaped sets 
    \[
        T_d = \{ x \in L_k \mid  x_i \leq \down{i}_i \text{ for all } i \in [3]\}
    \quad \text{ and } \quad 
        T_u = \{ x \in L_k \mid  \up{i}_i \leq x_i \text{ for all } i \in [3]\}
    \]
    and observe that $S = T_u \cap T_d$ by definition.
    If no pair $\up{i}, \down{i}$ is in third configuration, then by the assumption $\dia(S)_i \leq 1$ for some $i \in [3]$, we have that either $T_u$ or $T_d$ consists of at most $3$ points (see also \Cref{fig:three_configurations}). Without loss of generality, assume that it is $T_u$ (the other case is symmetric). We can also assume that $\up{1}_2 \leq \up{2}_2$ (without loss of generality, since $|T_u| \leq 3$ implies that there must exist such a pair, see \Cref{fig:three_configurations}) . 
    Now if we also had $\up{2}_1 \leq \up{1}_1$, the two points would be in first configuration. Thus, assume from now on $\up{2}_1 > \up{1}_1$. Then, we must have $\up{2}_3 \leq \up{3}_3$ (since $|T_u| \leq 3$, see~\Cref{fig:three_configurations}). Again, $\up{2}_2 \geq \up{3}_2$ would imply a first configuration, so we assume $\up{2}_2 < \up{3}_2$. But this then implies $\up{3}_1 \leq \up{1}_1$, which means that the three points are in second configuration. 
\end{proof}

\subsection{Third Configuration \texorpdfstring{$\implies$}{Implies} Progress} 
\label{sec:third_configuration}

Next, we argue that we can also find an upward or downward point from a third configuration in time $\bigO(\log N)$.

\begin{figure}[htb]
    \centering
    \includegraphics[width=0.4\linewidth]{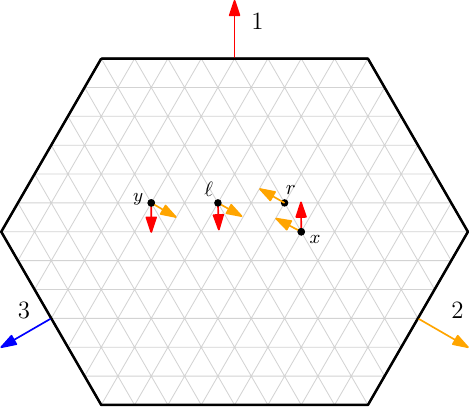}
    \caption{An example of the procedure that we use to deal with two points $x, y$ in third configuration. The point $r$ was initialized as the unique point with $r_1 = y_1$ and $r_2 = x_2 - 1$ while $\ell$ was initialized to be $y$ but then updated after one binary search step. Within $\bigO(\log N)$ steps, we will either find an upward or downward point, or shrink the interval between $\ell$ and $r$ to only contain two points, which again implies an upward or downward point, as desired.}
    \label{fig:third_configuration}
\end{figure}

\begin{lemma}[Third Configuration $\implies$ Progress]
\label{lemma:third_configuration}
    Given a pair of points $x, y \in L_k$ in third configuration, we can find an upward point in $L_{\geq k}$ or a downward point in $L_{\leq k}$ in time $\bigO(\log N)$.
\end{lemma}
\begin{proof}
    Assume without loss of generality that $x$ is $1$-upward and $y$ is $1$-downward. Recall that we must have $x_1 \leq y_1 \leq x_1 + 1$. Observe that there has to be a coordinate $j \in \{2, 3\}$ satisfying $y_j < x_j$ since both $x$ and $y$ belong to $L_k$ and must be distinct. Without loss of generality, we will assume $j = 2$. From $x, y \in L_k$ we thus deduce $x_3 \leq y_3$. \Cref{fig:third_configuration} shows a sketch of the situation.
    
    Since $x$ is $1$-upward, we must have $F(x)_2 \leq x_2$. Observe that $F(x)_2 = x_2$ would imply that $\LUB(x, y) = (y_1, x_2, y_3)$ is an upward point: Indeed, $F(\LUB(x, y))_1 \geq \LUB(x, y)_1$ follows from combining $y_1 \leq x_1 + 1$ with $x$ being $1$-upward, whereas 
    $F(\LUB(x, y))_2 \geq \LUB(x, y)_2$ follows from $F(x)_2 = x_2$ and $F(\LUB(x, y))_3 \geq \LUB(x, y)_3$ follows from $y$ being $1$-downward. Thus, assume from now on $F(x)_2 < x_2$. Analogously, we can assume $F(y)_2 > y_2$ since otherwise, $\GLB(x, y)$ is a downward point. 

    Next, consider an arbitrary point $q \in L_k$ with $q_1 = y_1$ and $y_2 \leq q_2 < x_2$. If we have $F(q)_1 \geq q_1$ and $F(q)_2 \geq q_2$, then $\LUB(q, y) = (y_1 = q_1, q_2, y_3)$ is an upward point. Conversely, $F(q)_1 < q_1$ and $F(q)_2 \leq q_2$ together imply that $\GLB(x, q) = (x_1 \geq q_1 - 1, q_2, x_3)$ is a downward point. The only remaining cases are that $q$ either satisfies 
    \begin{itemize}
        \item[(1)] $F(q)_1 \geq q_1$ and $F(q)_2 < q_2$, 
        \item[(2)] or $F(q)_1 < q_1$ and $F(q)_2 > q_2$.
    \end{itemize}
    Our algorithm thus works as follows: We maintain two points $\ell, r \in L_k$ with $\ell_1 = r_1 = y_1$ and $y_2 \leq l_2 < r_2 < x_2$ such that $r$ is of type (1) and $\ell$ is of type (2). We initialize $\ell = y$ and let $r$ be the unique point with $r_1 = y_1$ and $r_2 = x_2 - 1$ (if $r$ is not of type (1) it implies and upward or downward point by our previous observation on points $q \in L_k$ with $q_1 = y_1$ and $y_2 \leq q_2 < x_2$). If we find any point that is not of type (1) or (2), we get an upward or downward point, as desired. Otherwise, we update $r$ or $\ell$ depending on whether the point is type (1) or (2). Thus, by proceeding in a binary search fashion, we need at most $\bigO(\log N)$ queries to either find an upward or downward point, or arrive at $\ell_2 + 1= r_2$ (and thus $\ell_3 = r_3 + 1$). In the latter case, either $\GLB(\ell, r) = (\ell_1 = r_1, \ell_2 = r_2 - 1, r_3 = \ell_3 - 1)$ is downward (if $F(r)_3 \leq r_3$) or $\LUB(\ell, r) = (\ell_1 = r_1, r_2 = \ell + 1, \ell_3 = r_3 + 1)$ is upward (if $F(r)_3 > r_3$). 
\end{proof}

\subsection{Initializing Search Space}

It remains to explain how we initially set up our search space. In particular, we need to explain how we can find initial points $\up{i}, \down{i} \in L_k$ satisfying $\up{i}_i \leq \down{i}_i$ for all $i \in [3]$ in $\bigO(\log N)$ time. We end up using a similar technique as in \Cref{lemma:third_configuration}.

\begin{figure}[htb]
    \centering
    \includegraphics[width=0.4\linewidth]{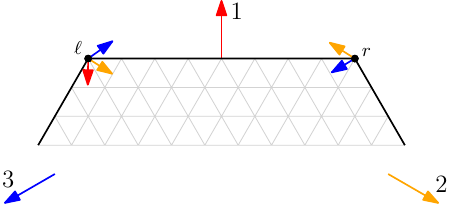}
    \caption{An example of the situation of finding an initial point $\down{1}$ in $L_k$, as described in \Cref{lemma:init}. The algorithm maintains $\ell$ and $r$ and proceeds in a binary search fashion. As long as no solution has been found, we will be able to maintain $F(\ell)_3 < \ell_3$ and $F(r)_2 < r_2$ while shrinking the interval between the two points.}
    \label{fig:init}
\end{figure}

\begin{lemma}[Initializing $\up{i}$ and $\down{i}$]
\label{lemma:init}
    For every $i \in [3]$, we can find an $i$-upward point $\up{i} \in L_k$ and an $i$-downward point $\down{i} \in L_k$ with $\up{i}_i \leq \down{i}_i$ in time $\bigO(\log N)$, or directly return an upward point in $L_{\geq k}$ or a downward point in $L_{\leq k}$.
\end{lemma}
\begin{proof}
    Without loss of generality, assume $i = 1$. We will explain how we can find $\down{1} \in L_k$ in time $\bigO(\log N)$. The case of finding $\up{1}$ is symmetric and the property $\up{1}_1 \leq \down{1}_1$ will be guaranteed since $\down{1}_1$ will be as large as it can get among all points in $L_k$ (and symmetrically, $\up{1}_1$ will be as small as it can get in $L_k$).

    Concretely, let $\ell \in L_k$ be the unique point with $\ell_1$ maximized and $\ell_2$ minimized. Observe that at least two of the three inequalities $F(\ell)_1 \leq \ell_1$, $F(\ell)_2 \geq \ell_2$, and $F(\ell)_3 \leq \ell_3$ must hold at $\ell$ since $\ell$ must lie on an edge of our 3D-Grid. Thus, if we had $F(\ell)_2 < \ell_2$, $\ell$ would be downward and we would be done. Similarly, it is not hard to check by case distinction that $F(\ell)_3 \geq \ell_3$ implies that $\ell$ is either $1$-downward, upward, or downward. Thus, we can assume from now on $F(\ell)_2 \geq \ell_2$ and $F(\ell)_3 < \ell_3$.

    We also define $r \in L_k$ to be the unique point with $r_1$ maximized and $r_3$ minimized. Analogously to the case of $\ell$, we can assume $F(r)_2 < r_2$ and $F(r)_3 \geq r_3$. An example of the situation is shown in \Cref{fig:init}. 

    Now let $q \in L_k$ be an arbitrary point with $q_1 = \ell_1 = r_1$. First consider the special case where $F(q)_1 > q_1$. This would mean that there are points in the 3D-grid with larger first coordinate, but they are not part of the levelset. This can only happen if $\ell = r$, which by our previous observations implies that we would already be done. Thus, we can assume $F(q)_1 \leq q_1$. 
    
    Next, observe that we have $\ell_2 \leq q_2 \leq r_2$ and $r_3 \leq q_3 \leq \ell_3$. If we have both $F(q)_2 \geq q_2$ and $F(q)_3 \geq q_3$, then $q$ is $1$-downward and we are done. Otherwise, we get that $q$ satisfies either $F(q)_2 < q_2$ or $F(q)_3 < q_3$. 

    We now proceed in a binary search fashion, updating $\ell$ to $q$ whenever $F(q)_3 < q_3$ and $r$ to $q$ whenever $F(q)_2 < q_2$. After at most $\bigO(\log N)$ steps, we will have either found a solution or arrive at $\ell_2 \leq r_2 \leq \ell_2 + 1$ and $r_3 \leq \ell_3 \leq r_3 + 1$. In the latter case, we deduce that $\GLB(\ell, r)$ is a downward point by using $F(\ell)_3 < \ell_3$ and $F(r)_2 < r_2$.
\end{proof}

\subsection{Putting Everything Together} 
\label{sec:putting_everything_together}

We conclude by putting all the pieces together. Given a levelset $L_k$ of an instance of 3D-\Tarski, we can use \Cref{lemma:init} to initialize the six bounding points that induce our search space. We then use \Cref{lemma:shrink_search_space_I} and \Cref{lemma:shrink_search_space_II} to iteratively shrink our search space $S$ and achieve $\dia(S)_i \leq 1$ for some $i \in [3]$. Finally, we check our six bounding points for one of the three configurations. From \Cref{lemma:implied_configuration}, we know that at least one of the three configurations must exist. Depending on the configuration, we then use \Cref{observation:first_config}, \Cref{observation:second_config}, or \Cref{lemma:third_configuration} to conclude with an upward point in $L_{\geq k}$ or a downward point in $L_{\leq k}$. Each of those steps takes at most $\bigO(\log N)$ time, yielding an overall runtime of at most $\bigO(\log N)$. Together with \Cref{lemma:high-level_algo}, we thus get an overall $\bigO(\log^2 N)$-time algorithm for 3D-\Tarski.

\section*{Acknowledgments.}

We want to thank Simon Weber and Mika Göös for valuable discussions, and anonymous reviewers for their detailed feedback.

\bibliographystyle{plainurl}
\bibliography{references.bib}

\end{document}